\documentclass[11pt]{article}

\usepackage{amssymb}
\usepackage{amsmath}
\usepackage{graphicx}
\usepackage{color}
\usepackage{fullpage}

\usepackage{algorithmic}
\usepackage[linesnumbered,ruled]{algorithm2e}

\begin{document}

\newtheorem{theorem}{Theorem}[section]
\newtheorem{lemma}{Lemma}[section]
\newtheorem{corollary}{Corollary}[section]
\newtheorem{claim}{Claim}[section]
\newtheorem{proposition}{Proposition}[section]
\newtheorem{definition}{Definition}[section]
\newtheorem{fact}{Fact}[section]
\newtheorem{example}{Example}[section]

\newcommand{\quod}{\hfill $\blacksquare$ \bigbreak}
\newcommand{\reals}{I\!\!R}
\newcommand{\property}{I\!\!P}
\newcommand{\np}{\mbox{{\sc NP}}}
\newcommand{\sing}{\mbox{{\sc Sing}}}
\newcommand{\con}{\mbox{{\sc Con}}}
\newcommand{\prob}{\mbox{Prob}}
\newcommand{\atm}{\mbox{{\sc ATM}}}
\newcommand{\hopn}{\hop_{\cN}}
\newcommand{\atmn}{\atm_{\cN}}
\newcommand{\cA}{{\cal A}}
\newcommand{\cO}{{\cal O}}
\newcommand{\cP}{{\cal P}}
\newcommand{\cC}{{\cal C}}
\newcommand{\C}{{\cal C}}
\newcommand{\cB}{{\cal B}}
\newcommand{\cG}{{\cal G}}
\newcommand{\cN}{{\cal N}}
\newcommand{\cU}{{\cal U}}
\newcommand{\cF}{{\cal F}}
\newcommand{\cT}{{\cal T}}
\newcommand{\hx}{\hat{x}}
\newcommand{\cS}{{\cal S}}

\newcommand{\EE}{$\epsilon$-\-en\-ve\-lope}
\newcommand{\EQ}{$\epsilon$-\-con\-quer}
\newcommand{\ED}{{\tt En\-ve\-lo\-pe-Dis\-co\-ve\-ry}}
\newcommand{\SC}{{\tt Safe\--Con\-quer}}
\newcommand{\CE}{{\tt Con\-fir\-med-\-E\-cho}}
\newcommand{\CT}{{\tt Co\-lor\&\-Tran\-smit}}
\newcommand{\AC}{{\tt As\-sign-\-Co\-lor}}
\newcommand{\CD}{{\tt Con\-fli\-ct-\-De\-te\-ction}}
\newcommand{\CR}{{\tt Con\-fli\-ct-\-Re\-so\-lu\-tion}}
\newcommand{\SGR}{{\tt Shifted Grid-Refinement}}
\newcommand{\ERRD}{{\tt Error-Detection}}
\newcommand{\WGR}{{\tt Witnessed Grid-Refinement}}

\newcommand{\DB}{$\Delta$-block}
\newcommand{\DBs}{$\Delta$-blocks}
\newcommand{\FDB}{$5\Delta$-block}
\newcommand{\FDBs}{$5\Delta$-blocks}

\newcommand{\UB}{{\tt Universal Broadcast}}
\newcommand{\CAB}{{\tt Company-Aware Broadcast}}
\newcommand{\DI}{{\tt Dense-1}}
\newcommand{\DII}{{\tt Dense-2}}

\newcommand{\MS}{\mathcal{S}}

\newcommand{\eps}{{\epsilon}}
\newcommand{\la}{{\lambda}}
\newcommand{\al}{{\alpha}}
\newcommand{\qed}{\hfill $\square$ \smallbreak}

\newcommand{\UDGI}{{\tt UDG1}}
\newcommand{\UDGII}{{\tt UDG2}}
\newcommand{\SYM}{{\tt SYM}}

\newenvironment{proof}{\noindent{\bf Proof:}}{\qed}

\def\lalto{\left \lceil}
\def\ralto{\right \rceil}
\def\lbasso{\left \lfloor}
\def\rbasso{\right \rfloor}
\def\D{{\Delta}}
\def\qed{\hfill$\Box$}

\baselineskip    0.19in
\parskip         0.1in
\parindent       0.0in

\bibliographystyle{plain}
\title{{\bf Advice Complexity of Treasure Hunt\\ in Geometric Terrains }}
\author{
Andrzej Pelc \footnotemark[1] \footnotemark[2]
\and
Ram Narayan Yadav \footnotemark[3]
}
\date{ }
\maketitle
\def\thefootnote{\fnsymbol{footnote}}

\footnotetext[1]{ \noindent
D\'epartement d'informatique, Universit\'e du Qu\'ebec en Outaouais, Gatineau,
Qu\'ebec J8X 3X7, Canada.
E-mail: {\tt pelc@uqo.ca} 
}
\footnotetext[2]{ \noindent  
Partially supported by NSERC discovery grant 2018-03899 and     
by the Research Chair in Distributed Computing at the
Universit\'e du Qu\'{e}bec en Outaouais. 
}
\footnotetext[3]{ \noindent
Department of Computer Science and Engineering,
Indian Institute of Information Technology, Dharwad, India.
E-mail : {\tt narayanram1988@gmail.com}  
}

\begin{abstract}
Treasure hunt is the task of finding an inert target by a mobile agent in an unknown environment. We consider treasure hunt in geometric terrains with obstacles.
Both the terrain and the obstacles are modeled as polygons and both the agent and the treasure are modeled as points. 
The agent navigates in the terrain, avoiding obstacles, and finds the treasure when there is a segment
of length at most 1 between them, unobstructed by the  boundary of the terrain or by the obstacles. The cost of finding the treasure is the length of the trajectory of the agent. 
We investigate the amount of information that the agent needs {\em a priori} in order to find the treasure at cost 
$O(L)$, where $L$ is the length of a shortest path in the terrain from the initial position of the agent to the treasure, avoiding obstacles.
Following the well-established paradigm of {\em algorithms with advice}, this information is given to the agent in advance as a binary string, by an oracle cooperating with the agent and knowing the whole environment:
in our case, the terrain,
the position of the treasure and the initial position of the agent. Advice complexity of treasure hunt is the minimum length of the advice string (up to multiplicative constants)
that enables the agent to find the treasure at cost $O(L)$.

We first consider treasure hunt in {\em regular} terrains which are defined as convex polygons with convex $c$-fat obstacles, for some constant $c>1$.
A polygon is $c$-fat if the ratio of the radius of the smallest disc containing it to the radius of the largest disc contained in it is at most $c$. For the class of regular terrains,
we establish the exact advice complexity of treasure hunt. We then show that advice complexity of treasure hunt for the class of arbitrary terrains (even for non-convex polygons without obstacles,
and even for those with only horizontal or vertical sides) is exponentially larger than for regular terrains.

\vspace*{1cm}

\noindent {\bf keywords:} mobile robot, treasure hunt, polygon, obstacle.

\vspace*{2cm}

\end{abstract}

\thispagestyle{empty}

\pagebreak

\section{Introduction}

\subsection{The background and the problem}
Treasure hunt is the task of finding an inert target by a mobile agent in an unknown environment. We consider treasure hunt in geometric terrains with obstacles.
This task has important applications when the terrain is dangerous or difficult to access for humans. Rescuing operations in mines contaminated or submerged by water are an example of situations
where a lost miner is a target that has to be found fast by a mobile robot, and hence the length of the robot's trajectory should be as short as possible.

We model the treasure hunt problem as follows. The terrain is represented by an arbitrary polygon $\cP _0$ with
pairwise disjoint polygonal obstacles
$\cP _1,..., \cP _k$, included in the interior of $\cP _0$, i.e., the terrain is $\cT=\cP _0 \setminus (\cP _1\cup \cdots \cup \cP _k)$.
We assume that the polygon  $\cP _0$ is closed (i.e., contains its boundary) and the polygons $\cP _1,..., \cP _k$ are open (i.e., do not contain their boundaries).
In this way the terrain  $\cT$ is a closed subset of the plane, i.e., it contains its entire boundary, which is the union of boundaries of all polygons.
It should be noted that the restriction to polygons is only to simplify
the description, and all our results can be generalized to the case where polygons 
are replaced by compact subsets of the plane homeotopic with a disc (i.e., without holes) 
and regular enough to have well-defined boundary length. 
The treasure is modeled as an inert interior point of the terrain. 

The mobile agent (robot) is modeled as a point starting inside the terrain 
and moving along a polygonal line inside it. It is equipped with a compass and a unit of length, and we assume that it has unbounded memory: from the computational point of view the agent is a Turing machine. 
The moves of the agent are of two types: {\em free} moves and {\em boundary} moves.
A free move is a move of the agent in the terrain along a segment of a chosen length in a chosen direction. Such a move may be interrupted if the agent hits the boundary of the terrain during its execution, 
and the agent becomes aware of this interruption. A boundary move is executed by an agent located on the boundary of  the terrain. Such a move is of the form: follow the boundary that you are on (this can be the boundary of any of the polygons ${\cal P}_i$)
in a chosen direction (there are two possible directions), either  at a chosen distance or until getting to a point with a given property.

The aim of treasure hunt is for the agent to {\em see} the treasure. We assume that the agent currently located at a point $p$ of the terrain sees all  points $q$  for which the segment $pq$
is entirely contained in $\cT$ and is of length at most 1. The cost of a treasure hunt algorithm on an instance is the length of the trajectory of the agent from its initial position until it sees the treasure.
We assume that the agent does not know the terrain nor the location of the treasure before starting treasure hunt.

We investigate the amount of information that the agent needs {\em a priori} in order to find the treasure at cost 
$O(L)$, where $L$ is the length of a shortest path in the terrain from the initial position of the agent to the treasure.
\footnote{Since we use the $O$-notation for functions with positive real  values that can be smaller than 1, it is important to give a precise definition that we will use: For two functions $f,g: \mathcal{R}^+ \longrightarrow \mathcal{R}^+$, we say that $f(x)$ is in $O(g(x))$ if there exists a positive constant $c$ such that $f(x) \leq cg(x)$, for all $x \in \mathcal{R}^+$.}
(Since the agent sees at distance 1, we have  $C\leq L \leq C+1$, where $C$ is the optimal cost of treasure hunt with full knowledge).
Following the well-established paradigm of {\em algorithms with advice} (see the subsection ``Related work''), this information is given to the agent in advance as a binary string, by an oracle
cooperating with the agent and knowing the whole environment:
in our case, the terrain,
the position of the treasure and the initial position of the agent. Advice complexity of treasure hunt is the minimum length of the advice string (up to multiplicative constants)
that enables the agent to find the treasure at cost $O(L)$. Advice complexity of a task can be considered to be a measure of its difficulty. 
Hence  our aim is to estimate the difficulty of treasure hunt in geometric terrains. It is well known that many algorithmic tasks become feasible or easier,
when the algorithm is supplied with a particular item of information, such as the size or diameter of the graph. However, the paradigm of algorithms with advice permits us to establish the minimum size of the information needed, regardless of its nature. Hence the measure of advice complexity is a quantitative approach to the knowledge provided to the algorithm, as opposed to the qualitative approach, studying the impact of knowing particular items of information, such as various numerical parameters of the problem. To the best of our knowledge, this is the first time that the advice complexity approach is applied to the problem of treasure hunt in the geometric setting.

Coming back to our application in the context of a miner lost in the mine, advice complexity of treasure hunt may be crucial. The miner knows the terrain, knows his/her position and knows the entrance to the mine. How to text as little information as possible to the rescuing team (time is precious) to allow a robot to reach the miner fast?

In order to formulate our results, we define the following parameter $\lambda$ called the {\em accessibility} of the treasure. $\lambda=\min(1/2,\rho)$, where $\rho$ is the largest radius, such that some disc with radius $\rho$ contains the treasure and is contained in the terrain. Since the treasure is located in an interior point of the terrain, we have $\lambda>0$. By definition, any disc of radius $\lambda$ containing the treasure and contained in the terrain has the property that the agent reaching any point of this disc can see the treasure.

%
\subsection{Our results}
We first consider treasure hunt in {\em regular} terrains which are defined as convex polygons with convex $c$-fat obstacles, for some constant $c>1$.
A polygon is $c$-fat if the ratio of the radius of the smallest disc containing it to the radius of the largest disc contained in it is at most $c$.
(For example, all regular convex polygons are 2-fat). For the class of regular terrains,
we establish the exact advice complexity of treasure hunt. 
For $L>\lambda$, we provide a treasure hunt algorithm working at cost $O(L)$ for all regular terrains, using advice of size $O( \log(L/\lambda))$, and we construct a class of regular terrains for which there is no treasure hunt algorithm working at cost $O(L)$ with advice of size $o( \log(L/\lambda))$. For $L\leq \lambda$, we construct a treasure hunt algorithm working at cost $O(L)$ for all regular terrains, without any advice.

In order to appreciate the strength of the tightness result for $L>\lambda$, notice that its positive part gives a concrete advice of size $O( \log(L/\lambda))$
(in our case indicating the approximate direction towards the treasure with respect to the initial position of the agent), and a treasure hunt algorithm of cost $O(L)$
for the class of regular terrains,  using this advice, while the negative part shows that 
no advice of size of smaller order, {\em regardless of its kind and meaning}, can help to accomplish treasure hunt at cost $O(L)$ in all regular terrains.

We then show that advice complexity of treasure hunt for the class of arbitrary terrains (even for non-convex polygons without obstacles,
and even for those with only horizontal or vertical sides) is exponentially larger than for regular terrains. Our negative result is even stronger: we construct terrains with treasure accessibilty 1/2 
for which advice complexity of treasure hunt can be a function of $L$ growing arbitrarily fast.

\subsection{Related work}

{\bf Treasure hunt.}
The problem of searching for a target by one or more mobile agents was investigated under many different scenarios.
The environment where the target is hidden may be a graph or a plane, and the search may be deterministic or randomized.
The book \cite{AG} surveys both the search for a fixed target and the related rendezvous problem, where the target and the searching agent are both mobile and
their role is symmetric: they cooperate to meet. This book is concerned mostly with randomized search strategies. In \cite{MP,TSZ} the authors studied relations between the problems of treasure hunt (searching for a fixed target) and rendezvous in graphs.  The authors of \cite{BCR} studied the task of finding a fixed point on the line and in the grid, and initiated the study of the task
of searching for an unknown line in the plane. This research was continued, e.g., in \cite{JL,La2}. In \cite{SF} the authors concentrated on game-theoretic aspects of
the situation where multiple selfish pursuers compete to find a target, e.g., in a ring. The main result of \cite{La} is an optimal algorithm to sweep a plane in order to locate an unknown fixed target, where locating means to get the agent originating at point $O$ to a point $P$ such that the target is in the segment $OP$. In \cite{FHGTM} the authors considered the generalization of the search problem in the plane to the case of several searchers.  Efficient search for a fixed or a moving target in the plane, under complete ignorance of the searching agent, was studied in \cite{Pe}.

{\bf Exploration of terrains.}
Exploration of unknown terrains by mobile robots is a subject closely related to treasure hunt.
A mobile agent has to see all points of the terrain, where seeing a point $p$ means either the 
existence of a segment between the current position of the agent and $p$ inside the terrain
(unlimited vision), or the existence of such a segment of length at most 1 (limited vision).
Most of the research in
this domain concerns the competitive framework, where the trajectory of the robot not knowing
the environment is compared to that of the optimal exploration algorithm having full knowledge.  

In~\cite{DKP98}, 
the authors gave a $2$-competitive algorithm for
rectilinear polygon exploration with unlimited vision. 
The case of non-rectilinear polygons (without obstacles) 
was also studied in~\cite{DKP91,HIKK01}
and a competitive algorithm was given in this case.

For polygonal environments with an arbitrary number of polygonal
obstacles, it was shown in~\cite{DKP98} that no competitive strategy
exists, even if all obstacles are parallelograms. Later, this result
was improved in~\cite{AKS02} by giving a lower bound in
$\Omega(\sqrt{k})$ for the competitive ratio of any on-line algorithm
exploring a polygon with $k$ obstacles. This bound remains
true even for rectangular
obstacles. Nevertheless, if the number of obstacles is bounded by a
constant $k$, then there exists a competitive algorithm with
competitive ratio in $O(k)$~\cite{DKP91}.

Exploration of polygons by a robot with limited vision has been studied,
e.g., in \cite{GB01,GBBS08,N92}. In \cite{GB01} 
the authors described an on-line algorithm with
competitive ratio $1+3(\Pi D/A)$, where $\Pi$ is a quantity depending
on the perimeter of the polygon, $D$ is the area seen by the robot, and $A$ is the area of
the polygon. 
In \cite{N92} the author studied exploration of the
boundary of a terrain with limited vision. The cost of exploration of arbitrary terrains with obstacles, both for limited and unlimited vision, 
was studied in \cite{CILP}.


Navigation in a $n\times n$ square room filled with rectangle obstacles
aligned with sides of the square
was considered in \cite{BBFY,BPFKRS,BRS,PY}. It was shown in \cite{BBFY}
that the navigation
from a corner to the center of a room can be performed with a
competitive ratio $O(\log n)$, only
using tactile information (i.e., the robot modeled as a point sees an obstacle only when it touches it). 
No deterministic algorithm can achieve a better
competitive ratio, even with
unlimited vision~\cite{BBFY}. For navigation between any pair of
points, there
is a deterministic algorithm achieving a competitive ratio of
$O(\sqrt{n})$~\cite{BRS}. No deterministic
algorithm can achieve a better competitive ratio~\cite{PY}. However,
there is a randomized approach performing
navigation with a competitive ratio of $O(n^{\frac{4}{9}}\log
n)$~\cite{BPFKRS}.


{\bf Algorithms with advice.}
The paradigm of algorithms with advice was developed mostly for tasks in graphs.
Providing arbitrary types of knowledge that can be used to increase efficiency of solutions to network problems 
 has been
proposed in \cite{AKM01,DP,EFKR,FGIP,FIP1,FIP2,FKL,FP,FPR,GPPR02,IKP,KKKP02,KKP05,MP,SN,TZ05}. This approach was referred to as
{\em algorithms with advice}.  
The advice is given either to the nodes of the network or to mobile agents performing some task in it.
In the first case, instead of advice, the term {\em informative labeling schemes} is sometimes used if different nodes can get different information.

Several authors studied the minimum size of advice required to solve
network problems in an efficient way. 
In \cite{FIP1}, the authors compared the minimum size of advice required to
solve two information dissemination problems using a linear number of messages. 
In \cite{FKL}, it was shown that advice of constant size given to the nodes enables the distributed construction of a minimum
spanning tree in logarithmic time. 
In \cite{DKM,EFKR}, the advice paradigm was used for online problems.
In \cite{FGIP}, the authors established lower bounds on the size of advice 
needed to beat time $\Theta(\log^*n)$
for 3-coloring cycles and to achieve time $\Theta(\log^*n)$ for 3-coloring unoriented trees.  
In the case of \cite{SN}, the issue was not efficiency but feasibility: it
was shown that $\Theta(n\log n)$ is the minimum size of advice
required to perform monotone connected graph clearing.
In \cite{IKP}, the authors studied radio networks for
which it is possible to perform centralized broadcasting in constant time. They proved that constant time is achievable with
$O(n)$ bits of advice in such networks, while
$o(n)$ bits are not enough. In \cite{FPR}, the authors studied the problem of topology recognition with advice given to the nodes. 
In \cite{DP}, the task of drawing an isomorphic map by an agent in a graph was considered, and the problem was to determine the minimum advice that has to be given to the agent for the task to be feasible. 
 Leader election with advice was studied in \cite{GMP} for trees, and in \cite{DiPe} for arbitrary graphs.
Graph exploration with advice was studied in \cite{BFU,GP} and treasure hunt with advice in graph environments was investigated in \cite{KKKS,MP}.

\section{Regular terrains}

In this section we give three results concerning treasure hunt in regular terrains. First we consider the case $L>\lambda$, where
 $L$ is the length of a shortest path in the terrain between the initial position of the agent and the location of the treasure, and  $\lambda$ is the accessibility of the treasure. In this case we construct a treasure hunt algorithm working in any regular terrain at cost $O(L)$, using advice of size $O(\log(L/\lambda))$  and we show that this size of advice is optimal for the class of regular terrains. Then we consider the case $L\leq \lambda$ and provide a a treasure hunt algorithm working at cost $O(L)$ for all regular terrains, without any advice.

In our algorithm for $L>\lambda$ we will need to convey advice that is conceptually a pair $(a_1,a_2)$, where $a_i$ are positive integers. However, by definition,  the advice has to be a single binary string, hence it is important to 
efficiently and unambiguously code such pairs as binary strings, so that the decoding be unambiguous as well and correctly restore the coded pair. This can be done as follows. A pair   $(a_1,a_2)$ can be viewed as a string over the 3-symbol alphabet with symbols $0,1$ and comma, where  $a_1$ and $a_2$ are represented in binary. Code any such string replacing 0 by 01, 1 by 10, and comma by 11. Denote the obtained binary string by $Code(a_1,a_2)$. It is clear that the pair $(a_1,a_2)$ can be unambiguously decoded from $Code(a_1,a_2)$ and that the length of $Code(a_1,a_2)$ is $O(\log (\max (a_1,a_2))$.

We start by describing the procedure $Walk(\gamma,x)$, that is at the core of both our algorithms for regular terrains. Its high-level idea is the following. Suppose that the initial position of the agent is $p$ and that $N$ is the half-line  starting at $p$ that forms angle $\gamma$ with the direction North. This half-line can possibly intersect some obstacles in the terrain. The aim of the procedure is to travel along the line $N$ circumventing the encountered obstacles in an efficient way, until the total trajectory travelled by the agent has length $x$. The agent walks along $N$ starting at $p$. When it hits an obstacle at a point $r$, the agent finds the other point  $r'$ of the intersection of $N$ with the perimeter of the obstacle (such a point $r'$ is unique by the convexity of the obstacle), using a version of the Cow Path walk (searching for an unknown point in the line, cf. \cite{BCR}) executed on the perimeter. Then the agent goes further along the line $N$, circumventing each encountered obstacle as above, until the total trajectory traveled by the agent has length $x$. We will show that if a point $q$ is in the line $N$ at distance $y$ from $p$ and outside of all obstacles then the smallest $x$ such that $Walk(\gamma,x)$ reaches $q$ is $O(y)$.
This is due to the fact that if a convex $c$-fat polygon is cut by a line then the smaller part of its perimeter between the cutting points is only $d$ times larger than the Euclidean distance between these points, where $d$ depends only on $c$. This is proved in Lemma~\ref{fat}. 

Suppose that the agent hits an obstacle $O$ at a point $r$ while moving along the line $N$. The agent starts searching for the other point  $r'$ of intersection of $N$ with the perimeter $R$ of the obstacle $O$ by using Procedure $CowPath(R,N,r)$ described below. $CowPath$ is a version of the Cow Path walk on the line transposed to the walk on the perimeter $R$. 

The agent identifies two directions of travelling on $R$, call these directions $dir_1$ and $dir_2$, defined as follows.
Let $a$ and $b$, with $a\leq b$, be the distances between $r$ and the two closest vertices of the polygon $O$ along the perimeter $R$.

If $a<b$ and $a\leq 1$ then $dir_1$ is towards the vertex at distance $a$. (In this case the agent can see this vertex).
Otherwise (either both distances $a$ and $b$ are equal, or they are both larger than 1), the agent could start in any of the two directions but it must be unambiguously defined. There are two cases.

If the point $r$ is one of the vertices of the polygon $O$ then $dir_1$ corresponds to the side adjacent to $r$ forming a smaller angle with direction North.  

If the point $r$ lies in the interior of a side $e$ of the polygon $O$ then $dir_1$ is defined as follows.
If $e$ is horizontal, then direction $dir_1$ corresponds to West. Otherwise, the direction $dir_1$ corresponds to the part of $e$ in the Northern half-plane defined by the horizontal line passing through $r$. In all cases, $dir_2$ is the other direction than $dir_1$. 

%
%

The agent walks on the perimeter $R$ of the obstacle $O$ starting at point $r$. Let $z=\min(1,a)$. First, it goes in direction $dir_1$ at distance $z$. Then the agent goes back to $r$ and goes in direction $dir_2$ until distance $2z$ is travelled or until it visits the point $r'$, whichever comes first. The agent swings in this way each time doubling the travelled distance until reaching point $r'$, when it finishes the execution of the procedure. Notice that the agent starts with distance $z$, rather than distance 1, as opposed to the classic Cow Path walk (that assumes that the target is at distance at least 1) because point $r'$ could be much closer to $r$ than 1 along $R$. Starting with distance $z$ is safe, as $z$ is smaller than the shorter path between $r$ and $r'$ along $R$.
 Algorithm \ref{CowPath} gives the pseudocode of Procedure $CowPath(R,N,r)$.

  \begin{algorithm}
\Begin
{
Compute $dir_1$, $dir_2$ and $z$\\
$i \leftarrow z$ \\
\While{the point $r'$ on the line $N$ other than $r$ is not visited}
{
Go on $R$ in direction $dir_1$ {\bf{until}} ($r'$ is visited) $or$ (distance $i$ is traveled)\\ 
\If{the point $r'$ is visited}{Exit}
\Else{Go back on $R$ to point $r$\\ $i \leftarrow 2i$}
Go on $R$ in direction $dir_2$ {\bf{until}} ($r'$ is visited) $or$ (distance $i$ is traveled)\\
\If{the point $r'$ is visited}{Exit}
\Else{Go back on $R$ to point $r$\\ $i \leftarrow 2i$}
}
}
\caption{Procedure $CowPath(R,N,r)$}
\label{CowPath}
\end{algorithm}

After finding the point $r'$ on the perimeter of the obstacle, the agent continues along the line $N$, using procedure $CowPath$ whenever an obstacle is hit,
until it travels a total trajectory of length $x$.
 Algorithm \ref{algo-walk} gives the pseudocode of procedure $Walk(\gamma, x)$. The algorithm is interrupted when the length of the total  trajectory traveled by the agent is $x$.

\pagebreak

\begin{algorithm}
\Begin
{
Let $N$ be the half-line starting at $p$ that forms angle $\gamma$ with the direction North\\ 
{\bf repeat}\\
\Begin
{
Go along the line $N$\\
\If{the perimeter $R$ of an obstacle is hit at point $r$}
{Call Procedure $CowPath(R,N,r)$}
}

}
\caption{Procedure $Walk(\gamma, x)$}
\label{algo-walk}
\end{algorithm} 

{\bf Remark.}
Since perimeters of obstacles are included in the terrain, if the line $N$ along which the agent travels intersects only the perimeter of an obstacle, this is not considered as hitting the obstacle, and the agent continues its travel along $N$.

The following geometric result will be crucial for the analysis of our algorithms.

\begin{lemma}\label{fat}
Let $c>1$ be a constant and let $P$ be a convex $c$-fat polygon. Consider any line $M$ cutting $P$ at points $a$ and $b$. Then, there is a constant $d$ such that the length $s$ of the smaller part of the perimeter of $P$ between points $a$ and $b$ is at most $d\cdot|ab|$.
\end{lemma}

\begin{proof}
Let $\alpha$ and $\beta$ be the interior angles at points $a$ and $b$, induced by the cut of the polygon $P$ by the line $M$. Let $|ab|=x$. We consider the following three cases:\\

\noindent
{\bf{Case 1.}} $\alpha=\pi-\beta$

In this case the sides of the polygon that are cut by the line $M$ are parallel, see Fig. $\ref{case-1}$.
Let $x'$ be the distance between the parallel lines containing these sides. In Fig. $\ref{case-1}$,  $x'=|ab'|$. By the definition of $r$, we have, $2r \leq |ab'|$. As the angle $\angle$ $ab'b$ is a right angle, we have  $2r \leq |ab'|\leq |ab|$. The length $s$ of the smaller part of the perimeter of $P$ between points $a$ and $b$ is at most $2\pi R$ because the entire convex polygon $P$ is contained in a circle of diameter $2R$. Hence, $s/|ab| \leq 2\pi R/2r \leq \pi\cdot c$. This proves the lemma in Case 1.\\

\begin{figure}[tp]
\centering
\includegraphics[scale=0.8]{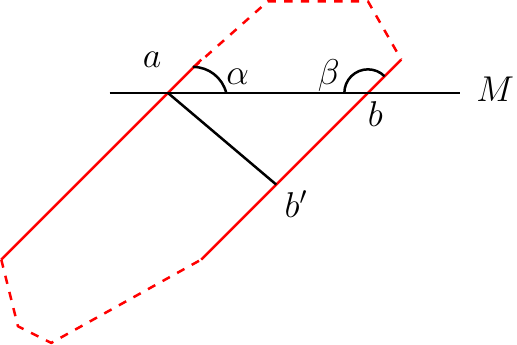}
\caption{Proof of Lemma \ref{fat} -- Case 1: $\alpha= \pi-\beta$}
\label{case-1}
\end{figure}

\noindent
{\bf{Case 2.}} $\alpha< \pi/2$ and $\beta\leq \pi/2$, or $\alpha> \pi/2$ and $\beta\geq \pi/2$

\begin{figure}[tp]
\centering
\includegraphics[scale=0.5]{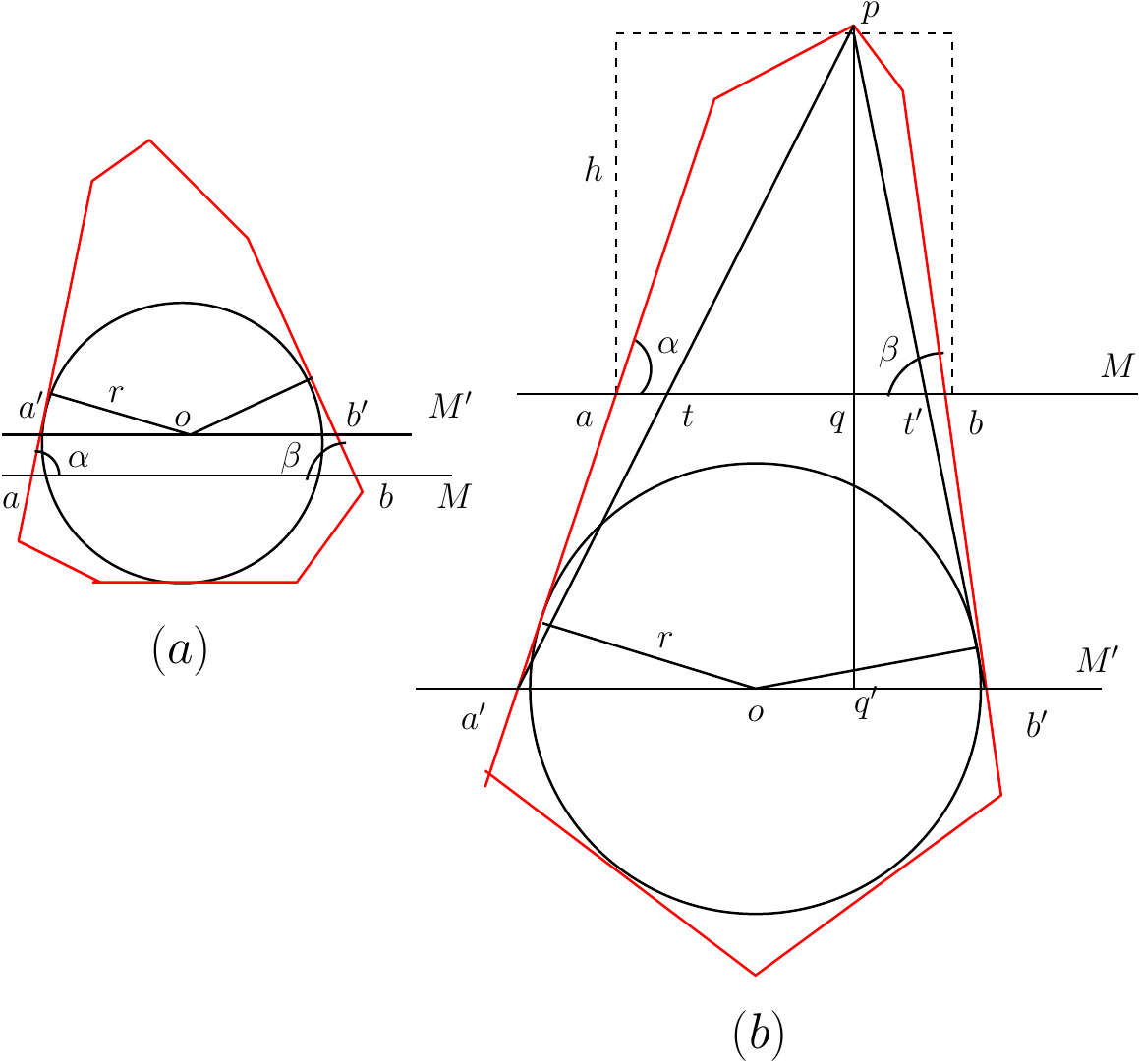}
\caption{Proof of Lemma \ref{fat} -- Case 2: $\alpha < \pi/2$ and $\beta \leq \pi/2$}
\label{case-2}
\end{figure}

It is enough to consider the first conjunction, as the second follows by interchanging the two parts of the perimeter of $P$ cut by $M$.
Consider the part of the perimeter of $P$ between points $a$ and $b$ corresponding to angles $\alpha$ and $\beta$ in the cut by line $M$. (This is the upper part of the perimeter in Fig. $\ref{case-2}$). There are two subcases.\\

\noindent
Subcase 2.1. The center of the largest circle contained in $P$ is in the upper part of $P$: Fig. \ref{case-2} (a).

In this case $x=|ab|\geq |a'b'|\geq 2r$, see  Fig. \ref{case-2} (a). The rest of the proof is as in Case 1.\\

\noindent
Subcase 2.2. The center of the largest circle contained in $P$ is in the lower part of $P$ : Fig. \ref{case-2} (b).

Let $p$ be the point in the upper part of the perimeter of $P$ farthest from the line $M$. Consider the line $M'$ parallel to $M$ containing the center of the
largest circle contained in $P$. Let $a'$ and $b'$ be the points in which the line $M'$ cuts the lines containing the sides of $P$ that are cut by $M$.
Let $y=|a'b'|$. Consider the triangle $a'b'p$. Let $t$ and $t'$ be the points where line $M$ cuts the sides $a'p$ and $b'p$ of this triangle, respectively.
Let $x'=|tt'|$. Hence $x'\leq x$. Let $q$ and $q'$ be the points of intersection with $M$ and $M'$, respectively, of the line perpendicular to $M$ and containing $p$.
Let $h=|pq|$ and $h'=|pq'|$. 
We have $h/x'=h'/y$. By definition of $R$ and $r$ we have $h'\leq 2R$ and $y \geq 2r$. Hence $h/x'\leq (2R)/(2r)\leq c$. Since $x'\leq x$, we have $h/x\leq c$.
Consider the rectangle (with dotted sides in Fig.  \ref{case-2} (b))  whose one side is the segment $ab$ and that contains $p$. The length of the upper part
of the perimeter of $P$ is at most $2h+x$. Hence this length is at most $(2c+1)x$, which proves the lemma in Case 2. 

\noindent
{\bf{Case 3.}} $\beta \neq \pi-\alpha$, $\alpha<\pi/2$ and  $\beta> \pi/2$

Let $\beta'=\pi-\beta$ be the exterior angle at point $b$ induced by the cut of polygon $P$ by the line $M$, see Fig. \ref{case-3}. Without loss of generality, we can assume $\alpha < \beta'$ in the rest of the proof of this case. For $\alpha >\beta'$, the proof is similar. Let $d$ be the endpoint of the side of the polygon $P$ containing point $a$, in the part of the perimeter corresponding to angle $\alpha$ (the upper part in Fig. \ref{case-3}). Let $y=|ad|$. Let $q$ be the point of intersection of the line containing points $a$ and $d$ and the line perpendicular to the side of polygon $P$ containing point $b$. Let $z=|dq|$ and $x''=|bq|$.  The segment $bq$ cuts the perimeter of the polygon $P$ at point $p$, see Fig. \ref{case-3}. Let $x'=|bp|$,  and let $t=|dp|$. Denote the angles $\angle$ $abq$ and $\angle$ $aqb$ by $\gamma$ and $\delta$ respectively. 

Since $\beta-\gamma = \pi/2$, we have $\gamma = \pi/2-\beta'$. As $\alpha < \beta'$, we have $\gamma = \pi/2-\beta' < \pi/2-\alpha$. Consider the triangle $aqb$.  Since $\alpha+\gamma+\delta = \pi$ and $\alpha+\gamma < \pi/2$, we have $\delta > \pi/2$. Hence, $y+z+x'' \leq 2x$. Since $t \leq z+(x''-x')$,
we have $y+t+x' \leq 2x$. Let the length of the part of the perimeter $P$, between the points $p$ and $b$ (clockwise in Fig. \ref{case-3}) be $s'$. By Case 2, there is a constant $k \geq 1$ such that $s'/x' \leq k$. We have $(y+t+s')/x \leq (y+t+kx')/x \leq k(y+t+x')/x$. Since $y+t+x' \leq 2x$, we have $k(y+t+x')/x \leq 2k$. This implies that $(y+t+s')/x \leq 2k$. This proves the lemma in Case~3. 
\begin{figure}[tp]
\centering
\includegraphics[scale=0.8]{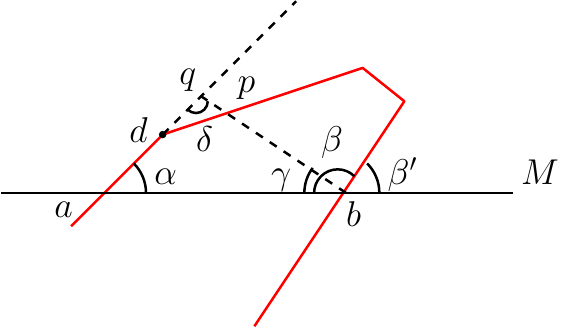}
\caption{Proof of Lemma \ref{fat} -- Case 3: $\beta \neq \pi-\alpha$, $\alpha<\pi/2$ and $\beta>\pi/2$}
\label{case-3}
\end{figure}
\end{proof}

We use Lemma \ref{fat} to prove the following result that will be used in the analysis of our algorithms.

\begin{lemma}\label{walk}
Let $N$ be the half-line  starting at $p$ that forms angle $\gamma$ with the direction North. 
Let $q$ be a point in the line $N$ at distance $y$ from $p$ and outside of all obstacles. Then the smallest $x$ such that $Walk(\gamma,x)$ reaches $q$ is $O(y)$.
\end{lemma}

\begin{proof}
Let $O_1, \dots ,O_t$ be the obstacles intersecting line $N$ between $p$ and $q$, in the order in which they are intersecting it, when the segment $pq$ is traversed from $p$ to $q$.
Let $r_1,r_1',r_2,r_2',\dots, r_t,r_t'$ be the points of intersection of $N$ with the perimeters of the obstacles, so that $r_i,r_i'$ belong to the perimeter of $O_i$.
Let $\pi_1$ be the segment $pr_1$, let $\pi_i$, for $i=2,\dots, t$, be the segment $r_{i-1}'r_{i}$, and let $\pi_{t+1}$ be the segment $r'_tq$. Let $\mu_i$, for $i=1,\dots,t$, be the shorter part of the perimeter of the obstacle $O_i$, between $r_i$ and $r_i'$. (In the case of equality, take any of the two parts).
The polygonal line $Q$ is defined as the concatenation of segments $\pi _1,\mu_1,\pi_2,\mu_2,\dots, \pi_t,\mu_t,\pi_{t+1}$ (cf. Fig. \ref{polyg}).  By Lemma \ref{fat}, the length of each part $\mu_i$ is at most $d$ times larger than the length of the segment $r_ir_i'$, for some constant $d\geq 1$, and thus the length of $Q$ is at most $d y$.

Consider the smallest $x$ such that $Walk(\gamma,x)$ reaches $q$.
Let $\mathcal{T}$ be the trajectory of the agent, starting from point $p$ and executing procedure $Walk(\gamma, x)$.
By definition, the length of $\mathcal{T}$  is $x$.
The competitive ratio of the Cow Path walk
on the line is 9 (cf. \cite{BCR}), and this holds regardless of the length of the first segment used, as long as it does not exceed the distance to the target (this length is 1 in the classic case and $z$ in our case). Hence the length of the part of the trajectory $\mathcal{T}$  corresponding to the perimeter of $O_i$ is at most 9 times larger than $\mu_i$. Hence the length of $\mathcal{T}$ is at most 9 times larger than the length of $Q$, and hence it is at most $9dy$. This proves the lemma. 
\end{proof}
\begin{figure}[tp]
\centering
\includegraphics[scale=0.8]{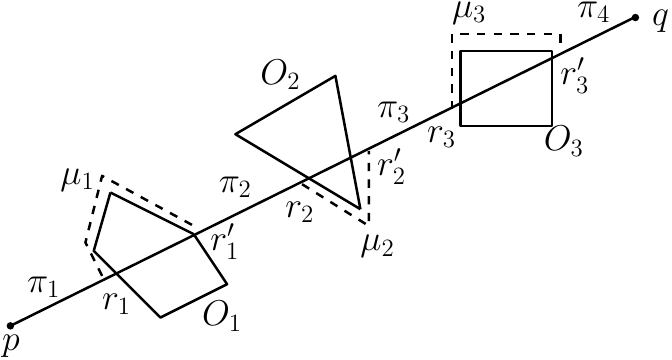}
\caption{Obstacles intersecting the line $N$ and the construction of the polygonal line $Q$}
\label{polyg}
\end{figure}

\subsection{The case $L>\lambda$}

In the case $L>\lambda$, the high-level idea of the advice and of the algorithm that uses it is the following. Let $k=\lceil7L/\lambda \rceil$. The oracle constructs $k$ half-lines $H_0,\dots,H_{k-1}$, starting at the initial position $p$ of the agent, such that $H_0$ goes North and the angle between consecutive half-lines is  $2\pi/k$. For $i=0,\dots, k-1$, the sector $S_i$ is defined as the set of points in the plane between $H_i$ and $H_{(i+1)\mod k}$, including $H_i$ and excluding $H_{(i+1)\mod k}$. Suppose that the treasure is in sector $S_j$. The oracle gives the string $Code(k,j)$ as advice. The algorithm decodes the couple $(k,j)$, and constructs 
angles $\gamma_1=2\pi j/k$ and $\gamma_2= 2\pi(j+1)/k$. Let $M_i$, for $i=1,2$, be the half-lines starting at $p$ forming angle $\gamma_i$ with direction North (these are the half-lines bounding sector $S_j$). Let $w'$ be the distance between $p$ and the closest point in any of the lines $M_1, M_2$ that does not belong to the terrain. Let $w=\min(1,w')$. The agent walks in phases, alternately on $M_1$ and $M_2$, using procedures $Walk(\gamma_1, 2^jw)$ and $Walk(\gamma_2, 2^jw)$ in phase $j=0,1,\dots$, and backtracking to $p$ in each phase using the reverse trajectory,  until it sees the treasure. The pseudocode of the algorithm is given in Algorithm \ref{far-treasure}.
It is interrupted as soon as the agent sees the treasure.

 \begin{algorithm}
\Begin
{
Decode the couple  $(k,j)$\\
Compute angles $\gamma_1=2\pi j/k$ and $\gamma_2= 2\pi(j+1)/k$\\
Find $w$\\
$j \leftarrow 0$\\
{\bf repeat}\\
\Begin{
$Walk(\gamma_1,2^jw)$\\
backtrack to $p$\\
$Walk(\gamma_2,2^jw)$\\
backtrack to $p$\\
$j \leftarrow j+1$\\

}
}
\caption{Algorithm FarTreasure}
\label{far-treasure}
\end{algorithm}

The following theorem is the main positive result of this section.

\begin{theorem}\label{positive}
For $L>\lambda$, Algorithm $FarTreasure$ accomplishes treasure hunt in any regular terrain at cost $O(L)$ with advice of size
$O(\log(L/\lambda))$.
 \end{theorem}
 
 In order to prove the theorem, we need the following lemmas.
 
 \begin{lemma}\label{advice-size}
The size of advice is $O(\log (\frac{L}{\lambda}))$.
\end{lemma}
\begin{proof}
As mentioned in the description of the string $Code(a_1,a_2)$, the length of this string is $O(\log (\max(a_1,a_2)))$.
We have $k=\lceil7L/\lambda \rceil$ and $j\leq k$, hence $\log k \in O(\log (L/\lambda))$ and $\log j \in O(\log (L/\lambda))$, which proves the lemma. 
\end{proof}

\begin{lemma}\label{cost}
For $L>\lambda$, Algorithm $FarTreasure$ correctly accomplishes treasure hunt in any regular terrain at cost $O(L)$.
\end{lemma}

\begin{proof}
The proof relies on the following geometric claim.

\begin{claim}\label{claim1}
Consider a sector  with angle $\alpha=2\pi/k$ formed by half-lines $M_1$ and $M_2$ starting at a point $p$. Let $q$ be a point in this sector at distance at most $L>\lambda$ from $p$. Then any circle of radius $\lambda$ containing point $q$ must intersect at least one of the half-lines $M_1$ or $M_2$ at distance at most $2L$ from point $p$. 
\end{claim}
In order to prove the claim, we first show that any circle of radius $\lambda$ containing point $q$ must intersect at least one of the half-lines $M_1$ or $M_2$. Suppose by contradicition that there
exists such a circle $C$ not intersecting any of these lines. $C$ has the center at a distance
$r \leq \lambda +L<2L$ from $p$. The largest circle containing $q$ with center at distance $r$ from $p$ not intersecting any of the half-lines $M_1$ or $M_2$ has the center $c$ on the bisectrix of angle $\alpha$. This center $c$ is at distance $x$ from both lines $M_1$ and $M_2$. By our assumption, $x\geq \lambda$.  The angle $\beta=\alpha/2$ between the bisectrix and any of the lines $M_1$ and $M_2$ corresponds to the arc of length $\rho=\beta r$ of a circle with radius $r$ centered at $p$. We have $x<\rho=\beta r= \pi r/k<\frac{\pi \cdot 2L}{7L/\lambda}<\lambda$. This is a contradiction that proves
that any circle of radius $\lambda$ containing point $q$ must intersect at least one of the half-lines $M_1$ or $M_2$.

It remains to show that some point of intersection of such a circle with at least one of the half-lines $M_1$ or $M_2$ is at distance at most $2L$ from $p$. The distance from $p$ to the closest such point must be smaller than $r<2L$, which concludes the proof of the claim.

Using the above claim, the lemma can be proved as follows. Let $\Delta$ be a circle of radius $\lambda$ containing the treasure and contained in the terrain.
Let $M_i$, for $i=1,2$, be the half-line starting at $p$ forming angle $\gamma_i$ with direction North (these are the half-lines bounding the sector containing the treasure). Since the treasure is at distance at most $L$ from $p$, it follows from the claim that there exists a point $q$ in $\Delta$ and situated in one of the lines $M_i$, at distance at most $2L$ from $p$. By Lemma \ref{walk}, 
the agent will get to point $q$ during the execution of the repeat loop of Algorithm $FarTreasure$ for the smallest $j$ such that $bL \leq 2^jw $, where $b$ is some constant.
Denote this smallest $j$ by $j_0$. At the point $q$ the agent sees the treasure. Hence the cost of Algorithm $FarTreasure$ is at most
$4w(1+2+\cdots +2^{j_0})\leq 8w2^{j_0}$. Since by definition of $j_0$ we have $w2^{j_0-1}<bL$, the cost of Algorithm $FarTreasure$ is at most $16bL$, which concludes the proof of the lemma.
\end{proof}

Now the proof of Theorem \ref{positive} is a direct consequence of Lemmas \ref{advice-size} and \ref{cost}.

We end this section by showing that the size $O(\log (L/\lambda))$  of advice used by Algorithm $FarTreasure$ is optimal for the class of regular terrains.
In order to prove this, we construct a class of regular terrains with treasure accessibility $\lambda$ for which any treasure hunt algorithm working at cost $O(L)$
must use advice at least $\frac{1}{2}\log (L/\lambda)$.

We start the construction by defining the {\em gadget} $G(o)$, for
any point $o$ in the plane. The gadget consists of 8 squares of side $x=3\lambda/2$, situated as follows (see Fig.
\ref{angle}). There are four squares $\sigma_N$, $\sigma_E$, $\sigma_S$, $\sigma_W$ whose centers are at distance $\lambda+x/2$ from point $o$, respectively, North, East, South and West from this point.
The remaining four squares are placed as follows. Let  $y=(2\lambda-x)/2$. 
The center of $\sigma_{NW}$ is West of the center of $\sigma_N$, at distance $x+y$ from it.
The center of $\sigma_{NE}$ is East of the center of $\sigma_N$, at distance $x+y$ from it.
The center of $\sigma_{SE}$ is East of the center of $\sigma_S$, at distance $x+y$ from it.
The center of $\sigma_{SW}$ is West of the center of $\sigma_S$, at distance $x+y$ from it.
Notice that the NW corner of $\sigma_{NW}$, the NE corner of $\sigma_{NE}$, the SE corner of $\sigma_{SE}$ and the SW corner of $\sigma_{SW}$
are corners of a square $S(o)$ of side $3x+2y=5\lambda$ which is the convex hull of the eight squares of the gadget.

\begin{lemma}\label{seeing}
An agent located outside of the square $S(o)$ cannot see the point $o$.
\end{lemma}
\begin{proof}
Let $a$ be the point North from $o$ at distance $\lambda$ from it. Let $b$ be the SE corner of $\sigma_N$. Let $\theta_2$ be the angle $\angle$ $aob$. Let $c$ be the SW corner of $\sigma_{NE}$. Let $d$ the NW corner of $\sigma_{NE}$. Let $\theta_1$ be the angle $\angle$ $cdb$, see Fig. \ref{angle}. 
Due to the symmetries of the gadget $G(o)$, the lemma follows from the inequality $\theta_2>\theta_1$. This inequality is proved as follows.
$\tan \theta_2=x/(2\lambda)$. $\tan\theta_1=y/x$. Hence 
$$\tan \theta_2=\frac{3\lambda/2}{2\lambda}=3/4>1/6=\frac{\lambda/4}{3\lambda/2}=\frac{(2\lambda-x)/2}{x}=y/x=\tan \theta_1.$$
Hence $\theta_2>\theta_1$.
\end{proof}

\begin{figure}[tp]
\centering
\includegraphics[scale=0.7]{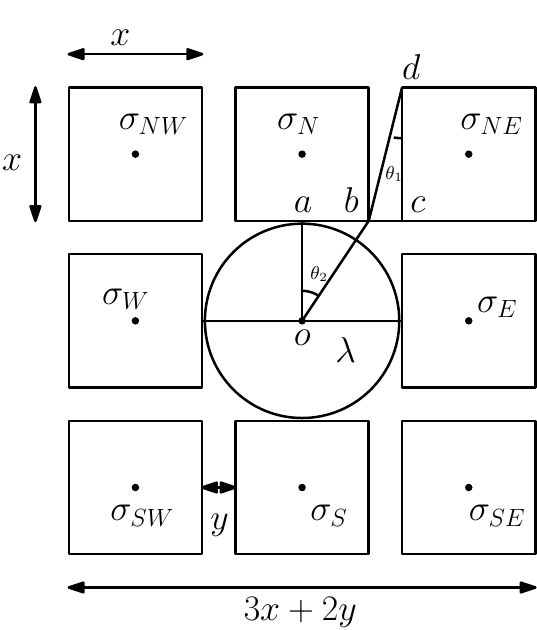}
\caption{The gadget $G(o)$}
\label{angle}
\end{figure}

\begin{theorem}\label{regular lb}
Let $\lambda$ be a real in the interval $(0,1/2]$.
There exist arbitrarily large reals $L$ and regular terrains with an initial position $p$ of the agent and the location $q$ of the treasure in it, such that $L$ is the length of a shortest path in the terrain between $p$ and $q$, $\lambda$ is the accessibility of the treasure, and the agent needs advice of size at least $\frac{1}{2}\log (L/\lambda)$ to accomplish treasure hunt at cost $O(L)$ in this terrain. 
\end{theorem}

\begin{proof}
We construct the regular terrain as follows. 
Consider a square with side $A=20k\lambda$, where $k$ is a positive integer. Let the initial position $p$ of the agent be the South-West corner of $S$, see Fig. \ref{lb-1}. Partition the square $S$ into four equal quadrants. Let $S'$ be the top-right quadrant. $S'$ has side of length $A/2$.
Partition $S'$ into $\frac{A^2}{100\lambda^2}$ square tiles
of side $5\lambda$. Tile rows are indexed $1,2,\dots$ from the North side of $S'$ going South, and tile columns, and are indexed $1,2,\dots$ from the West side of $S'$ going East. 
Now, place the gadget $G(o)$ contained in square $S(o)$ of side $5\lambda$ (see Fig. \ref{angle}) in every other tile of odd-indexed tile rows in $S'$ 
(see Fig. \ref{lb-1}, where these tiles are shaded: we will call them shaded tiles in the sequel). The distance between any two shaded tiles is at least $5\lambda$. 
At least 1/4 of all tiles are shaded, so the number of such tiles is at least $\frac{A^2}{400\lambda^2}$. This finishes the description of the terrain.
Since all obstacles are squares, the terrain is regular.  Notice that if the treasure is placed in the center of any shaded tile then $\lambda$ is the accessibility of the treasure.

We prove the theorem by contradiction. Suppose that the minimum size of advice to accomplish treasure hunt at cost $O(L)$ in any terrain from the above described class is less than $\frac{1}{2}\log (L/\lambda)$. Consider all possible locations of the treasure in the center of any of the shaded tiles.
 Notice that the length $L$ of a shortest path in the terrain from $p$ to any such center is in the interval $[\sqrt{2}A/2,\sqrt{2}A/2+A]$.

Using less than $\frac{1}{2}\log (L/\lambda)$ bits, we have at most $\sqrt{L/\lambda}$ different advice strings. By the pigeonhole principle, there is a subset  $\Sigma$ of at least $\frac{A^2/(400\lambda^2)}{\sqrt{L/\lambda}}$ shaded tiles, such that the location of the treasure at their center corresponds to the same advice string. By Lemma \ref{seeing}, in order to see the treasure located at the center of a shaded tile, the agent must get to this tile. 
Since the minimum distance between any two shaded tiles is at least $5\lambda$, 
any trajectory $T$ of the agent that enables it to accomplish treasure hunt when the treasure is located at the center of some tile in the set $\Sigma$, must have length at least 
$\frac{(A^2/(400\lambda^2)}{\sqrt{L/\lambda}} \cdot 5\lambda  =  \frac{ A^2}{80 \sqrt{\lambda L}}$. 
Since $L$ is in the interval $[\sqrt{2}A/2,\sqrt{2}A/2+A]$, the length of $T $\ is at least $c\frac{L^{3/2}}{\sqrt{\lambda}}$ for some positive constant $c$. Since $\lambda \leq 1/2$, the cost of treasure hunt is at least $cL^{3/2}$ and hence cannot be linear in $L$, which gives a contradiction.
\begin{figure}[tp]
\centering
\includegraphics[scale=1.0]{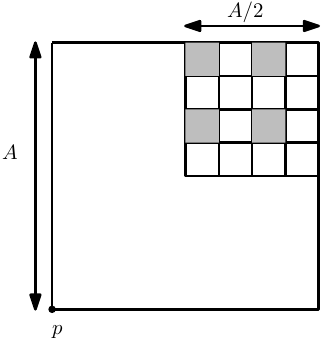}
\caption{Regular terrrain with initial position $p$ of the agent}
\label{lb-1}
\end{figure}
\end{proof}

Theorems \ref{positive} and \ref{regular lb} imply that for $L>\lambda$,  the size $O(\log (L/\lambda))$ of advice used by Algorithm $FarTreasure$ is sufficient
to find the treasure in all regular terrains at cost $O(L)$ and cannot be improved.

\subsection{The case $L\leq \lambda$}

In the case $L\leq \lambda$, we show an algorithm working without any advice and finding the treasure at cost $O(L)$.
The high-level idea of the algorithm is the following. Let $L_i$, for $i=0,\dots, 11$, be the half-lines starting at the initial position $p$ of the agent and forming angle $\gamma_i=\pi i/6$ with direction North. Let $w'$ be the distance between $p$ and the closest point in any of the lines $L_i$ that does not belong to the terrain. Let $w=\min(1,w')$. The agent walks in phases, in a round-robin fashion on lines $L_0,L_1,\dots,L_{11}$, using procedures $Walk(\gamma_0, 2^jw)$,  $Walk(\gamma_1, 2^jw)$, ..., $Walk(\gamma_{11}, 2^jw)$  in phase $j=0,1,\dots$, and backtracking to $p$ in each phase using the reverse trajectory,  until it sees the treasure. The pseudocode of the algorithm is given in Algorithm \ref{close-treasure}.
It is interrupted as soon as the agent sees the treasure.

 \begin{algorithm}
\Begin
{
Find $w$\\
$j \leftarrow 0$\\
{\bf repeat}\\
\Begin{
\For{$i:=0$ to $11$}
{$Walk(\pi i/6,2^jw)$\\
backtrack to $p$\\
}
$j \leftarrow j+1$\\

}
}
\caption{Algorithm CloseTreasure}
\label{close-treasure}
\end{algorithm}

 The following result proves the correctness and estimates the cost of Algorithm $CloseTreasure$.
 
 \begin{theorem}\label{close treasure}
 For $L\leq \lambda$, Algorithm $CloseTreasure$ accomplishes treasure hunt in any regular terrain at cost $O(L)$ with no advice. 
 \end{theorem}
 
 \begin{proof}
 We  consider two cases.
Let $\Delta$ be a circle of radius $\lambda$ containing the treasure and contained in the terrain.

Case 1. $\lambda/9 \leq L\leq \lambda$.

We first prove the following claim, similar to Claim \ref{claim1}

 \begin{claim}\label{claim2}
Consider a sector  with angle $\alpha=\pi/6$ formed by half-lines $M_1$ and $M_2$ starting at a point $p$. Let $q$ be a point in this sector at distance at most $L$ from $p$, where $\lambda/9 \leq L\leq \lambda$. Then any circle of radius $\lambda$ containing point $q$ must intersect at least one of the half-lines $M_1$ or $M_2$ at distance at most $10L$ from point $p$. 
\end{claim}
In order to prove the claim, we first show that any circle of radius $\lambda$ containing point $q$ must intersect at least one of the half-lines $M_1$ or $M_2$. Suppose by contradicition that there
exists such a circle $C$ not intersecting any of these lines. $C$ has the center at a distance
$r \leq \lambda +L<2\lambda$ from $p$. The largest circle containing $q$, with center at distance $r$ from $p$, not intersecting any of the half-lines $M_1$ or $M_2$ has the center $c$ on the bisectrix of angle $\alpha$. This center $c$ is at distance $x$ from both lines $M_1$ and $M_2$. By our assumption, $x\geq \lambda$.  The angle $\beta=\alpha/2=\pi/12$ between the bisectrix and any of the lines $M_1$ and $M_2$ corresponds to the arc of length $\rho=\beta r$ of a circle with radius $r$ centered at $p$. We have $x<\rho=\beta r= \pi r/12\leq\frac{\pi \cdot 2\lambda}{12}<\lambda$. This is a contradiction that proves
that any circle of radius $\lambda$ containing point $q$ must intersect at least one of the half-lines $M_1$ or $M_2$.

It remains to show that some point of intersection of such a circle with at least one of the half-lines $M_1$ or $M_2$ is at distance at most $10L$ from $p$. The distance from $p$ to the closest such point must be smaller than $r\leq \lambda +L \leq 10L$, which concludes the proof of the claim.

 Let $Q$ be the half-line starting at $p$ and containing the treasure.  Suppose that the angle $\gamma$ between the direction North and $Q$ satisfies
 $\pi i/6 \leq \gamma <\pi(i+1)/6$. Hence the treasure is in the sector formed by half-lines $L_i$ and $L_{i+1}$.
Let $M_1=L_i$ and $M_2=L_{i+1}$. Since the treasure is at distance at most $L$ from $p$, it follows from Claim \ref{claim2} that there exists a point $q'$ in $\Delta$ and situated in one of the lines $M_i$, at distance at most $10L$ from $p$. 

Case 2.  $L<\lambda/9$.

In this case we will use the following claim.

\begin{claim}\label{bound}
Consider positive reals $d,\lambda$, such that $\lambda \geq 9d$. Let $PR$ be a line segment of length $\lambda +d$ and let $S$ be a point in this segment such that $|PS|=d$ and $|SR|=\lambda$. Let $M$ be a half-line starting from point $P$, forming an angle $\pi/6$ with the segment $PR$. Let $Q$ be the point in $M$ closest to $P$ such that $|QR|=\lambda$. Then $|PQ| \leq 1.2 \cdot |PS|$.
\end{claim}

\begin{figure}[tp]
\centering
\includegraphics[scale=1.0]{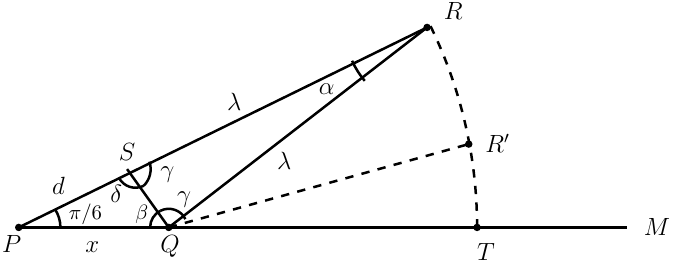}
\caption{Illustration of the proof of Lemma \ref{bound} }
\label{long-range}
\end{figure}

Let $|PQ|=x$ and $|PS|=d$. Denote the angle $\angle PSQ$ by  $\delta$, the angle $\angle SQR$ by $\gamma$, the angle $\angle QRS$ by $\alpha$ and the angle $\angle PQS$ by $\beta$, see Fig. \ref{long-range}. 
 
By the sine rule applied to the triangle $PQR$, we have $\frac{\sin (\pi/6)}{\sin (\beta+\gamma)}=\frac{\lambda}{\lambda +d}$. Since $\lambda/(\lambda +d) \geq 0.9$, we have $\sin(\beta+\gamma) \leq \sin (\pi/6)/0.9 \leq 0.56$. We have arcsin 0.56 = 0.594.
Let $A=(\pi-0.594) \geq 2.54$. Since $Q$ is the point in $M$ closest to $P$ such that $|QR|=\lambda$, we have $(\beta+\gamma) > \pi/2$. Thus, by definition of $A$ we have $(\beta+\gamma) \geq A$. Using the triangle $PQR$, we have $\alpha = \pi-\pi/6-(\beta+\gamma) \leq \frac{5\pi}{6}-A$. Since the triangle $SQR$ is isoceles, we have $\gamma=(\pi-\alpha)/2$. Since $\alpha \leq \frac{5\pi}{6}-A $, we have $\gamma \geq \frac{\pi}{12}+\frac{A}{2}$. By definition, we have $\delta = \pi-\gamma$. Since $\gamma \geq \frac{\pi}{12}+\frac{A}{2}$, we have $\delta \leq \frac{11 \pi}{12}-\frac{A}{2}$. Using the triangle $PQS$, we have $\beta = (\pi-\frac{\pi}{6}-\delta)=\frac{5\pi}{6}-\delta$. Since $\delta \leq \frac{11 \pi}{12}-\frac{A}{2}$, we have $\beta \geq \frac{A}{2}-\frac{\pi}{12} \geq \frac{2.54}{2}-\frac{\pi}{12}\geq 1$.

By the sine rule applied to the triangle $PQS$, we $\frac{x}{d}=\frac{\sin \delta}{\sin \beta} \leq \frac{1}{\sin\beta}$. Since $\pi/2\geq \beta \geq 1$, we have 
$\sin \beta \geq \sin 1$ and thus
$\frac{x}{d}\leq \frac{1}{\sin 1} \leq 1.2$. This proves the claim.

Let $d$ be the distance between point $p$ and the treasure. We have $d\leq L<\lambda/9$.  
The distance between $p$ and the center of $\Delta$ is at most $\lambda +d$. 
Let $M_1=L_i$ and $M_2=L_{i+1}$ be half-lines such that the center of $\Delta$ is located in the sector formed by $M_1$ and $M_2$.

First suppose that the center of $\Delta$ is located on one of the lines $M_1$ or $M_2$ at distance exactly $d+\lambda$ from $p$. Since the angle between lines $M_1$ and $M_2$ is $\pi/6$, it follows from Claim \ref{bound} (with $P$ replaced by $p$) that circle $\Delta$ intersects the other line $M_i$ at distance $x\leq 1.2 \cdot d$ from $p$. Hence, if the center of $\Delta$ is located anywhere in the sector formed by $M_1$ and $M_2$ at some distance at most $\lambda +d$ from $p$, the circle $\Delta$ intersects both  lines $M_i$ at distance $x\leq 1.2 \cdot d$ from $p$. This is due to the fact that in the situation depicted in Fig.\ref{long-range}
the length of the segment $PQ$ is a decreasing function of the length of the segment $PR$ and of the angle between segments $PR$ and $PQ$.

Hence in both cases we showed the existence of a point $q'$ in $\Delta$ and situated in one of the lines $M_i$, at distance at most $10L$ from $p$. 
By Lemma \ref{walk}, 
the agent will get to point $q'$ during the execution of the repeat loop of Algorithm $CloseTreasure$ for the smallest $j$ such that $gL \leq 2^jw $, where $g$ is some constant.
Denote this smallest $j$ by $j_0$. At the point $q'$ the agent sees the treasure. Hence the cost of Algorithm $CloseTreasure$ is at most
$24w(1+2+\cdots +2^{j_0})\leq 48w2^{j_0}$. Since by definition of $j_0$ we have $w2^{j_0-1}<gL$, the cost of Algorithm $CloseTreasure$ is at most $96gL$, which concludes the proof of the theorem in both cases.
 \end{proof}
 
{\bf Remarks.}

1. Our algorithms $FarTreasure$ and $CloseTreasure$ can be merged into a single treasure hunt algorithm preserving the features of the two algorithms in the respective cases: with no advice use algorithm $CloseTreasure$ and with a non-empty advice string use algorithm $FarTreasure$.

2. The threshold  for $L$ in comparison to $\lambda$ separating the cases that yield treasure hunt with advice $O(\log (L/ \lambda))$ and treasure hunt with no advice is somewhat arbitrary. It is easy to see that whenever $L\leq c\lambda$, for some positive (even very large) constant $c$, a treasure hunt algorithm
at cost $O(L)$ with no advice can be designed by suitably increasing the number of half-lines $L_i$ along which procedure $Walk$ is executed in Algorithm $CloseTreasure$, i.e., decreasing the angle  $\pi /6$ between consecutive half-lines. This would result in increasing the constant hidden in the $O(L)$ cost bound. Using Algorithm $FarTreasure$ for $\lambda< L \leq c\lambda$ makes this hidden constant lower at the expense of a constant number of bits of advice.

\section{Arbitrary terrains}

In this section we show that the advice complexity of treasure hunt is dramatically larger for the class of arbitrary terrains than for that of regular terrains. We show that the size of advice required for treasure hunt at cost $O(L)$ in some non-convex polygons, even without obstacles and even with all sides horizontal or vertical, is exponentially larger than that for regular terrains. In fact, we will show that this difference may be even more significant.

\begin{theorem}
For arbitrarily large integers $A$, 
there exists a class ${\cal C}(A)$ of  (non-convex) polygons $\cal P$ without obstacles, with an initial position $p$ of the agent and the location $q$ of the treasure in each of these polygons, so that the length $L$ of a shortest path between $p$ and $q$ in $\cal P$ is $\Theta(A)$, the accessibility of the treasure is $\lambda=1/2$, and the smallest size of advice required for treasure hunt at cost $O(L)$ in all polygons of this class,  is at least linear in $L$.
\end{theorem}

\begin{proof}
Fix an integer $A>8$. We construct the class ${\cal C}(A)$ of (nonconvex) polygons ${\cal P}_i$ as follows. Consider the square $S$ of side length $A$ with vertical and horizontal sides. Remove from the square $S$ two rectangles with horizontal side of length $(A-1)/2$ and vertical side of length 1: one of these removed rectangles has its upper-left corner at the upper-left corner of $S$ and the other has its upper-right corner at the upper-right corner of $S$ (see Fig. \ref{non-convex}).
In the third quarter of the height of the square $S$ (counting from the upper horizontal side) we remove vertical stripes of width $x=1/2^A$ and height $A/4-x$, and two horizontal stripes with upper side at height $A/2$ of the square,  of height $x$, and lengths respectively $2(i-1)x$, for the left horizontal stripe and $A-(2i-3)x$, for the right horizontal stripe (see Fig. \ref{non-convex}). Let $k=A/(2x)=A\cdot 2^{A-1}$. In each of the polygons ${\cal P}_i$ there are $k$ vertical corridors of width $x$, $k-1$ of them closed from above at the height $A/2+x$ (counting from the upper horizontal side of the square $S$), and one corridor open, exactly the $i$th corridor counting from the left.

\begin{figure}[tp]
\centering
\includegraphics[scale=0.75]{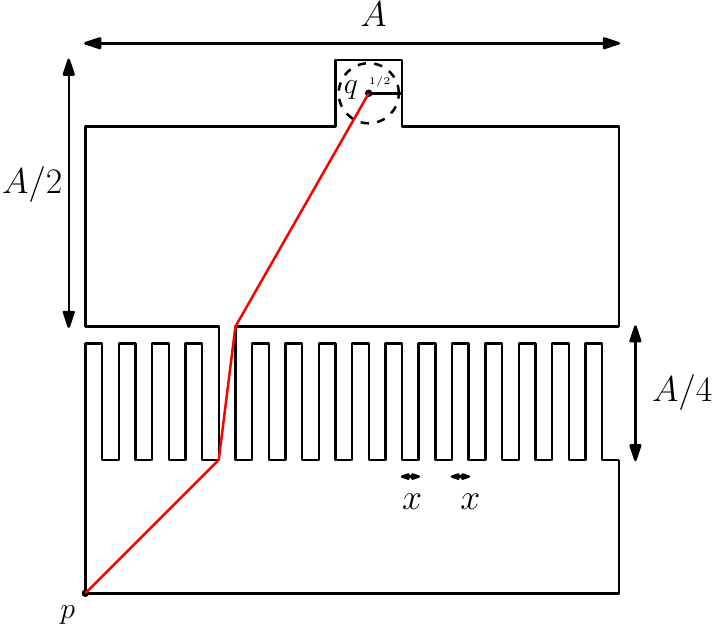}
\caption{A non-convex terrain (without obstacles) with the initial position $p$ of the agent and the location $q$ of the treasure in it. The accessibility of the treasure is $\lambda=1/2$.}
\label{non-convex}
\end{figure}

The initial position $p$ of the agent is the lower-left corner of the square $S$, and the location $q$ of the treasure is at distance 1/2 from the upper horizontal side
of $S$ and in the middle between the two vertical sides of $S$ (see Fig. \ref{non-convex}). Thus the accessibility of the treasure is $\lambda=1/2$ and the length $L$ of a shortest path between $p$ and $q$ in $\cal P$ satisfies the inequalities $A/2 <L<5A/2$, and hence $L\in \Theta(A)$. 

We prove the theorem by contradiction.
Suppose that the size of advice is at most $A/2$. Hence there are at most $2^{A/2}$ distinct pieces of advice. The number of polygons in the class  $\cal C(\cal A)$ is $k=A\cdot 2^{A-1}$. By the pigeonhole principle there are at least $y=k/2^{A/2}=A\cdot 2^{A/2-1}$ polygons in the class $\cal C(\cal A)$ to which corresponds the same advice $\alpha$.  Let ${\cal P}_{i_1},\dots , {\cal P}_{i_y}$ be these polygons. Consider the trajectory of the agent corresponding to advice $\alpha$.
This trajectory must enter each of the corridors $i_1,\dots ,i_y$, counted from the left, at the height at least $A/4 +1/2$, counting from the bottom side of the square $S$. This means that the agent must enter each of these corridors at depth at least 1/2 from the beginning of the corridor. At any lower point in a corridor, the agent cannot see if the corridor is open or closed (because it can see only at distance at most 1), and hence not going deeper in one of the corridors $j\in\{ i_1,\dots ,i_y\}$ would  preclude it from seeing the treasure, if the actual polygon is ${\cal P}_j$. Hence the trajectory of the agent corresponding to advice $\alpha$ must have a length of at least $2 \cdot (1/2) \cdot y=y$, in the case when the actual polygon is ${\cal P}_j$, where the last visited corridor has index $j$. 
However,  $y=A\cdot 2^{A/2-1}$ is not in $O(A)$ and hence not in $O(L)$, which is a contradiction. This contradiction proves that the size of advice must be larger than $A/2$, and hence it must be in $\Omega(L)$.
 \end{proof}
 
 {\bf Remark.} By replacing $x=1/2^A$ in the above proof by $1/f(A)$, for any faster growing function $f(A)$ (for example $f(A)=2^{2^A}$) we could get the lower bound
 $\Omega (\log f(L))$ on the required size of advice (instead of just $\Omega(L)$), and hence show an arbitrarily large difference between advice complexity of treasure hunt in arbitrary vs. regular terrains.

\section{Conclusion}

Using advice complexity as a measure of the difficulty of a task, we established that treasure hunt in the class of arbitrary terrains is dramatically more difficult than in the class of regular terrains.
A natural intermediate class of terrains is that of convex polygons with arbitrary convex obstacles (not necessarily $c$-fat). It remains open what is the advice complexity of treasure hunt in this class.


A problem related to treasure hunt  is that of finding a shortest path in a terrain. What is the advice complexity of this problem, i.e., what is the smallest advice that the agent needs in order to find a path of length exactly $L$ (rather than of length $O(L)$) to the target? Unfortunately, this is not a good formulation, as no finite advice could permit the agent to solve this problem, even in the empty plane.
Intuitively, the advice would have to convey the exact direction to the target, which cannot be done with a finite number of bits, as the target is a point. (This simple observation can be easily formalized).
An attempt to relax the task by requiring the (exact) shortest path not to hit the target but to ``see it'', i.e.,  get at distance 1 from it, must still fail for the same reason. It seems that a reasonable formulation of the shortest path problem in a terrain, in the context of advice complexity, has to relax the term ``shortest''. For example, the relaxation could be up to an {\em additive} constant (rather than up to a {\em multiplicative} constant, as we did, requiring cost $O(L)$). More precisely, the following problem remains open. What is the best complexity of advice sufficient to solve the treasure hunt problem in a terrain
(again, the agent has to see the treasure), at cost $L+O(1)$? Since the agent can see at distance 1, this version of treasure hunt is equivalent to solving the shortest path problem up to an additive constant.  
Similarly as in this paper, it would be interesting to find whether the difficulty of this problem (measured by advice complexity) varies for different classes of terrains.

\end{document}